\documentclass[a4paper,twocolumn,unpublished]{quantumarticle}
\pdfoutput=1

\usepackage[numbers,sort&compress]{natbib}

\usepackage{amsmath, amssymb,graphicx,bbm}
\usepackage{amsthm}
\usepackage{color}
\usepackage{physics}
\usepackage{docmute}
\usepackage{soul} 
\usepackage{ulem}
\usepackage{bbold}
\usepackage{nicefrac}
\usepackage{diagbox}
\usepackage{umoline}
\newtheorem{theorem}{Theorem}[section]
\newtheorem{lemma}[theorem]{Lemma}

\newcommand{\com}[1]{\left[#1\right]}

\newcommand{\orderof}[1]{\mathcal{O}(#1)} 

\newcommand{\EqDef}{\stackrel{\mathrm{def}}{=}}

\newcommand{\Eq}[1]{Eq.~(\ref{#1})}
\newcommand{\Fig}[1]{Fig.~\ref{#1}}
\newcommand{\iRef}[1]{Ref.~\cite{#1}}
\newcommand{\Refs}[1]{Refs.~\cite{#1}}
\newcommand{\App}[1]{Appendix~\ref{#1}}

\newcommand{\Nspam}{N_{spam}}

\begin{document}

\title{Optimal short-time measurements for Hamiltonian learning}

\author{Assaf Zubida}

\author{Elad Yitzhaki}

\author{Netanel H. Lindner}

\author{Eyal Bairey}
\email{baeyal@gmail.com}

\affiliation{Physics Department, Technion, 3200003, Haifa, Israel}

\begin{abstract}
Characterizing noisy quantum devices requires methods for learning the underlying  quantum Hamiltonian which governs their dynamics. Often, such methods compare measurements to simulations of candidate Hamiltonians, a task which requires exponential computational complexity. Here, we propose efficient measurement schemes based on short-time dynamics which circumvent this exponential difficulty. We provide estimates for the optimal measurement schedule and reconstruction error, and verify these estimates numerically. We demonstrate that the reconstruction requires a system-size independent number of experimental shots, and identify a minimal set of state preparations and measurements which yields optimal accuracy for learning short-ranged Hamiltonians. Finally, we show how grouping of commuting observables and use of Hamiltonian symmetries improve the accuracy of the Hamiltonian reconstruction.
\end{abstract}

\maketitle

\paragraph*{Introduction.} 
Recovering the unknown Hamiltonian of a quantum system is becoming an increasingly important task for modelling quantum materials \cite{Kwon2020, Wang2020}, characterizing quantum states \cite{Li2008,Turkeshi2018, Giudici2018, Zhu2018,Dalmonte2018,Kokail2020, Rattacaso2021OptimalStates} and engineering quantum devices \cite{Boulant2003, Innocenti2020SupervisedDesign, BenAv2020DirectQubit,  Carrasco2021TheoreticalVerification}. Whereas benchmarking protocols can extract the average error rates of quantum devices \cite{Magesan2011ScalableProcesses}, Hamiltonian learning provides a detailed characterization of their real-time dynamics, allowing to identify physical sources of errors \cite{Shulman2014SuppressingEstimation, Sheldon2016, Sundaresan2020, Eisert2020QuantumBenchmarking}. %As these devices grow, they raise the need for Hamiltonian learning approaches which require both experimental and computational resources scaling favorably with system size. 

Many Hamiltonian learning methods are based on measurements of time-dependent observables in quantum states evolved over time \cite{Han2021TomographyLearning, Sone2017, Zhang2014, Zhang2015, DiFranco2009, Burgarth2009, Samach2021LindbladProcessor}. However, as time progresses, the dynamics of many-body Hamiltonians become exponentially complex, posing a computational challenge when attempting to find the Hamiltonian which matches the observed behaviour. To deal with this exponential complexity, some approaches assume partial control of the quantum system to be characterized \cite{Wang2015, Valenti2019, Krastanov2019StochasticVariables, Valenti2021ScalableDynamics}, or employ an additional trusted quantum simulator \cite{Granade2012, Wiebe2014, Wiebe2014a, Wiebe2015, Wang2017}. Here the complexity refers both to the number of measurements and the classical computation involved in the learning procedure. 

%Recent works discuss an alternative scheme for dynamical Hamiltonian learning which is based on on energy conservation \cite{Bentsen2019IntegrableCavity, Li2019}. These works look for a local observable which does not change with time. In the absence of other local symmetries, the only conserved observable is the Hamiltonian. Each initial state or measurement time gives a single linear equation for the Hamiltonian, such that $\order{n}$ suffice to uniquely identify it. The equations remain computationally tractable for long evolution times, increasing their sensitivity at a fixed system size. In our short-time approach, each measurement settings yields not one but $\order{n}$ local equations for the Hamiltonian, providing an advantage at large system sizes.

A promising approach which avoids this complexity learns Hamiltonians from short-time dynamics \cite{Shabani2011, DaSilva2011}. This approach learns short-ranged Hamiltonians by preparing various product states, and measuring the time-derivative of different observables \cite{DaSilva2011}. The short-time dynamics are linear in the Hamiltonian parameters, allowing a computationally-efficient reconstruction. The rapid experimental preparation of product states provides an advantage over approaches which learn the dynamics from their steady states \cite{Qi2017, Chertkov2018, Greiter2018, Rudinger2015, Kieferova2017, Kappen2018, Bairey2019, Bairey2019a, Dumitrescu2019HamiltonianSystems,  Anshu2020, Evans2019ScalableLearning}. %\enote{Cutbacks?} Despite the promise of the short-time approach to Hamiltonian learning \cite{DaSilva2011}, optimal measurement protocols and performance guarantees have not been analyzed. In particular, protocols for estimating the optimal measurement time and for choosing initial state and measurement pairs have remained unexplored. 

Here we propose optimized measurement protocols for learning short-ranged Hamiltonians from short-time dynamics. We analyze the reconstruction error in a protocol that prepares random product states and measures in random local bases, as in recent randomized measurement schemes \cite{Li2019, Elben2019StatisticalStates, Evans2019ScalableLearning, Huang2020PredictingMeasurements}. We then develop a derandomized protocol that guarantees the same reconstruction accuracy with a fixed set of initial states and measurements. For a given state, the measurements of this protocol characterize all geometrically-local reduced density matrices using a fixed set of measurement bases, adapting overlapping tomography methods \cite{Cotler2020QuantumTomography, Garcia-Perez2020PairwiseSystems, Bonet-Monroig2020NearlyStates} to the special case of lattices. Our analysis shows that the total number of measurements required to recover a short-ranged Hamiltonian to a fixed accuracy is in fact system-size independent. Importantly, it provides an accurate numerical estimate for the optimal measurement time based on a rough prior guess of the Hamiltonian. 

%We provide numerical estimates for the reconstruction error and optimal measurement time in these protocols that do not rely on an accurate knowledge of the Hamiltonian. We propose two protocols that minimize sampling errors through the simultaneous measurability of commuting observables, inspired by recent  randomized \cite{Li2019, Elben2019StatisticalStates, Evans2019ScalableLearning, Huang2020PredictingMeasurements} and structured \cite{Cotler2020QuantumTomography, Garcia-Perez2020PairwiseSystems, Bonet-Monroig2020NearlyStates} protocols. In particular, for each initial state, our structured measurment protocol characterizes the reduced density matrices of all neighbors on a lattice using a fixed number of measurement bases. %Our numerical simulations verify the system-size independence of the reconstruction in these measurement schemes, and estimate the total number of measurements required for learning various realistic Hamiltonians.

\paragraph*{Problem setting.}
The goal of the scheme we analyze is to recover short-ranged Hamiltonians from their short-time dynamics \cite{DaSilva2011}. We say that a closed quantum system on a lattice is governed by a $k$-ranged Hamiltonian if each term in the Hamiltonian is contained in a cube of side length $k$. The procedure for learning a $k$-ranged Hamiltonian $H$ consists of (i) initializing the system in different product states, (ii) evolving the initial states for a short-time under $H$, and (iii) measuring the evolved states in various bases. Our goal is to reconstruct $H$ to a given accuracy using a minimal number of shots $N$, defined as the total number of experiments, from state preparation to measurement. Both the number of shots and the classical computation should scale favourably with the size of the system.

Any $k$-ranged Hamiltonian $H$ may be expanded in a basis $\lbrace{S_j \rbrace}$ for the space of $k$-ranged operators, 
\begin{equation}
    H=\sum_j c_j S_j,
\end{equation}
such that the Hamiltonian is determined by the coefficients $c_j$. The reconstruction of $H$ then boils down to finding the coefficient vector $\vec{c}$. For concreteness, we describe the reconstruction algorithm for spin $\nicefrac{1}{2}$ systems, where $S_j$ are products of Paulis acting on contiguous sites. While we describe it for closed systems, the algorithm can also recover the dynamics of a subsystem embedded in a larger system as in \iRef{Bairey2019}, as well as Markovian dissipative dynamics as in \Refs{Bairey2019a, Dumitrescu2019HamiltonianSystems}.

\paragraph*{Algorithm.}
The dynamics of an observable $A$ in an initial state $\ket{\psi}$ is governed by Ehrenfest's equation, 
\begin{equation} \label{Ehrenfest}
  \partial_t \ev{A}_\psi=i\ev{[H,A]}_\psi =i\sum_j c_j \ev{[S_j,A]}_\psi,  
\end{equation}
where $\ev{A}_{\psi} = \ev{\psi|A|\psi}$. For a known initial state $\ket{\psi}$ at $t=0$, \Eq{Ehrenfest} yields an  equation for the Hamiltonian coefficients $c_j$. We can get a set of equations using different pairs of initial states $\{\ket{\psi_\alpha}\}$ and observables ${\{
A_\beta\}}$. Denoting
\begin{equation} \label{K_element}
    b_{(\alpha,\beta)} \EqDef \partial_t \ev{A_\beta}_{\psi_\alpha}, \hspace{0.5cm}  K_{(\alpha,\beta),j} \EqDef i\ev{[S_j,A_\beta ]}_{\psi_\alpha},
\end{equation}
where we treat the pair $(\alpha,\beta)$ as a super index, \Eq{Ehrenfest} becomes
\begin{equation}
    b_{(\alpha, \beta)} = \sum_j K_{(\alpha, \beta), j} c_j,
\end{equation}
yielding a set of linear equations
\begin{equation} \label{linear_reconstruction}
\vec b=K\vec c.
\end{equation}
With sufficiently many initial states and observables, $K$ becomes full rank, and the Hamiltonian $\vec c$ is uniquely given by the solution $\vec{c}=K^+\vec{b}$, where $K^+=(K^TK)^{-1}K^T$ is the pseudo inverse of $K$. 

As can be seen in \Eq{linear_reconstruction}, the matrix $K$ maps Hamiltonians to observable dynamics. The matrix does not depend on the unknown Hamiltonian, and can be calculated in advance given a set of initial states and observables. The observable dynamics $\vec b$ contain information about $H$, and must be estimated from measurements to recover $H$. Each entry $b_{(\alpha, \beta)}$ is a derivative which can be approximated by a forward finite difference method,
\begin{equation} \label{finite_difference}
    \partial_t \ev{A}_{\psi}=\frac{\ev{A }_{\psi (\delta t)}-\ev{A}_{\psi (0)}}{\delta t}+\order{\delta t}.
\end{equation}
Starting with a known initial state, the only term we need to measure is $\ev{A}_{\psi(\delta t)}$, which we can evaluate by measuring $\ket{\psi(\delta t)}$ in a basis compatible with $A$, as we now explain.

To construct a set of equations, we initialize the system with random Pauli eigenstates and measure by projecting on random Pauli bases. We initialize
\begin{equation} \label{random_product_states}
    \ket{\psi(0)}=\prod_i \ket{\phi_i},
\end{equation} where each $\phi_i$ is chosen uniformly from one of the six $\pm1$ eigenstates of $X,Y$ and $Z$. We propagate $\ket{\psi}$ under $H$ for a time $\delta t$. We then measure each site of $\ket{\psi(\delta t)}$ randomly in one of the 3 Pauli bases $X,Y$ or $Z$. Namely, we measure $R^{\dagger} \ket{\psi (\delta t)}$ in the standard ($Z$) basis, where $R = \prod_i R_i$ and $R_i$ are chosen uniformly from $R_y (\nicefrac{\pi}{2}), R_x (\nicefrac{\pi}{2})$, and $\mathbbm{1}$. We call each choice of state preparation and measurement basis a SPAM setting, and denote the number of different SPAM settings used for reconstruction by $\Nspam$.

In a system with $n$ sites, each SPAM setting potentially contributes $2^n-1$ equations to the set (\ref{linear_reconstruction}), corresponding to $\partial_t \ev{A}_{\psi}$ for commuting Pauli observables $A$. For instance, consider a measurement in the $XZZX$ basis, which allows to estimate the single-site observables $\ev{X_1}$, $\ev{Z_2}$ and so on at $\ket{\psi(\delta t)}$. Moreover, it allows to estimate the two-site correlators $\ev{X_1 Z_2}$ as well as higher-order correlators. Fixing a set of $\Nspam$ different SPAM settings yields a system (\ref{linear_reconstruction}) with up to $\approx \Nspam \cdot 2^n$ equations. In practice, we only estimate observables up to a given range $k_{\beta}$, corresponding to $\approx n \cdot 2^{k_{\beta} - 1}$ equations per SPAM setting ($2^{k_{\beta} - 1}$ correlators extending from each site, up to boundary corrections). Collecting a total of $N$ shots distributed equally among the $\Nspam$ settings allows to estimate each of these equations to accuracy $\sim \frac{1}{\sqrt{N/\Nspam}} = \sqrt{\frac{\Nspam}{N}}$.

%let's look at an example of an experimental settings (fig 1), to understand which observables it can give us. The initial state of the system is a randomly chosen one of the 6 Pauli product state for each site, then propagate under the Hamiltonian for a time $\delta t$, and measure each site in an one of the 3 Pauli bases randomly. Assuming a 1D system of $n$ sites, we get from that experiment: $n$ $1-local, n-1$ $2-local, n-2$ $3-local...$ observables in the time $\delta t$.

%\ednote{Our way of fixing the sets of initial states $\{\ket{\psi_\alpha}\}$ and observables ${\lbraceA_\beta\rbrace}$ [\textit{\textbf{transition to SPAM}}?], is by random experimental setting. In each run we choose randomly one of the 6 Pauli product state for each site as the initial state of the system, propagate under the Hamiltonian for a time $\delta t$, choose randomly one of the 3 Pauli bases for each site and measure the system in that base. }

%In that method there are a still some free variable that may affect the reconstruction error, in that paper we try to optimize them:

This scheme raises a few questions, which we now address:

% \begin{figure}  \includegraphics[width=8.6cm]{fig1.pdf} 
%   \protect
%   \caption{Short-time measurements for learning local Hamiltonians.
% To learn a Hamiltonian $H$ we initiate each qubit randomly at one of the 6 Pauli eigenstates, evolve the system for a short-time $\delta t$ under $H$, and measure each qubit randomly in one of the 3 Pauli bases. The reconstruction error decreases with the number of times $N$ we repeat this process.} 
%  \label{fig1}
% \end{figure}

\begin{enumerate}
\item What is the optimal measurement time $\delta t$?
\item How many experiments are needed to recover a short-ranged Hamiltonian to a given accuracy?
\item How does the number of required experiments depend on the system size?
\item Which observables should we estimate from each measurement setting?
\end{enumerate}

\paragraph*{Estimating the reconstruction error and optimal measurement time.} We start by simulating our reconstruction algorithm on spin chains with random $2$-ranged interactions, using single-site observables. To this end, we generated Hamiltonians with nearest neighbor Paulis $S_j$, 
\begin{equation} \label{H_rand}
    H_{random}=\sum_j c_j S_j,
\end{equation}
where $S_j \in \lbrace \sigma_i \otimes \sigma_{i+1} \rbrace$ form a basis for all  nearest neighbor Paulis, with  $\sigma _i\in\{\mathbb{1}_i,X_i,Y_i,Z_i\}$ a Pauli operator on site $i$. The coefficients were drawn from a Gaussian distribution with zero mean and unit variance $c_j\sim G(0,1)$, setting the time scale for what follows. We initialized each system in $\Nspam=648$ random product states [\Eq{random_product_states}], and evolved them for time $\delta t$. We measured each evolved state $\nicefrac{N}{\Nspam}$ times in a random measurement basis, and estimated the expectation values of all single-site Paulis $\ev{A_\beta}_{\psi\left(\delta t\right)}$ compatible with that basis from the measurement results. %To simulate the estimation error of each observable from $N/\Nspam$ shots, we added Gaussian error with mean zero and standard deviation $\left( N/\Nspam \right)^{-\nicefrac{1}{2}}$. 
We then reconstructed the Hamiltonian coefficients $\vec{c}$ from \Eq{linear_reconstruction}, approximating the time-derivative of each observable using finite differences [\Eq{finite_difference}]. 

Reconstructing a random Hamiltonian using various evolution times $\delta t$ and a fixed measurement budget $N=10^6$, the reconstruction error was minimized at an optimal measurement time $\delta t^* \approx 0.02$ (\Fig{fig1}, left). We define the reconstruction error as the relative error of the reconstructed coefficients, 
\begin{equation} \label{reconstruction_error}
    \Delta \EqDef \frac{\norm{\delta \vec{c}}}{\norm{\vec{c}}}.
\end{equation}
Away from the optimal time, the reconstruction error increased significantly. %if we put this paragraph right after the algorithm, maybe we can cut it down
Repeating the reconstruction with various measurement budgets $N$, the optimal measurement time $\delta t^*$ scaled as $N^{-\nicefrac{1}{4}}$ (\Fig{fig1}, right, dots). Can we predict the optimal measurement time $\delta t^*$ and its reconstruction error $\Delta^*$?

\begin{figure}  \includegraphics[width=8 cm]{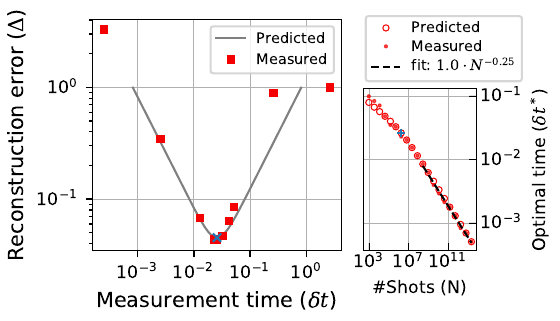} 
  \caption{Tradeoff between statistical and systematic errors determines optimal measurement time.
Left: reconstruction error $\Delta$ [\Eq{reconstruction_error}] of a random $2$-ranged Hamiltonian as a function of the measurement time $\delta t$ for a fixed measurement budget $N=10^6$. While the statistical error in the finite difference method decreases with $\delta t$, the systematic error increases with $\delta t$, dictating an optimal measurement time. A prior estimate of the reconstruction error (gray line), calculated using a noisy guess of the Hamiltonian [\Eq{guess}],  predicts the optimal time (blue 'x'). Time is measured in units of the variance of the Hamiltonain couplings [see \Eq{H_rand}]. Right: optimal measurement time $\delta t^*$ as a function of the number of shots $N$ (dots measured, circles predicted). The dashed line is a fit to the asymptotic scaling $\delta t^*\sim N^{-\gamma}$ with $\gamma = {1/4}$.
The reported values average over 10 different realizations of the statistical estimation error for the observables; errorbars are smaller than the markers.} 
 \label{fig1}
\end{figure}

The accuracy of the Hamiltonian reconstruction from \Eq{linear_reconstruction} is determined by the estimation accuracy of $\vec{b}$. In any experiment, instead of the exact $\vec{b}$ we can only obtain a noisy estimate $\vec{b} + \delta \vec{b}$. The estimation error $\delta \vec{b}$ corrupts the recovered coefficients, given by $\vec{c} + \delta \vec{c} = K^+ \left( \vec{b} + \delta \vec{b} \right)$. This leads to an error in the reconstructed Hamiltonian, such that the deviation between the true and the recovered coefficients is given by $\delta \vec{c}= K^+ \delta \vec{b}$. 

%The error in estimating $\delta \vec{b}$ from \Eq{finite_difference} arises due to two factors: the finite measurement budget and the finite-time approximation of the derivative. The finite sampling of $\frac{\ev{A(\delta t)}}{\delta t}$ leads to statistical error which decreases with the number of samples and with  time, scaling as $\delta \vec{b}_{stat} \sim \frac{1}{\sqrt{N} \delta t}$. The finite-time approximation of the derivative leads to a systematic error which increases with time, scaling as $\delta \vec{b}_{sys} \sim \delta t$. The optimal measurement time is determined by a tradeoff between those two factors. The total estimation error $\delta \vec{b} = \delta \vec{b}_{stat} + \delta \vec{b}_{sys}$ is minimized when they are both of comparable magnitude, which occurs at an optimal measurement time scaling as $\delta t ^* \sim N^{-\nicefrac{1}{4}}$. The estimation error at the optimal measurement time scales similarly as $\delta \vec{b}^* \sim N^{-\nicefrac{1}{4}}$, leading to reconstruction error scaling as $\Delta \sim N^{-\nicefrac{1}{4}}$.

The error $\delta \vec{b}$ in estimating the time-derivatives from \Eq{finite_difference} arises due to two factors: the finite measurement budget (number of experimental shots) and the finite-time approximation of the derivative. The finite sampling of $\frac{\ev{A(\delta t)}}{\delta t}$ leads to statistical error which decreases with the number of experiments and measurement time, $\norm{\delta \vec{b}_{stat}} \sim \frac{1}{\sqrt{N} \delta t}$ (for a fixed $\Nspam$). The finite-time approximation of the derivative leads to a systematic error which increases with time, $\norm{\delta \vec{b}_{sys}} \sim \delta t$. The optimal measurement time is determined by a tradeoff between those two factors. The total estimation error $\delta \vec{b} = \delta \vec{b}_{stat} + \delta \vec{b}_{sys}$ is minimized when they are both of comparable magnitude $\frac{1}{\sqrt{N} \delta t} \sim \delta t$, which occurs at an optimal measurement time scaling with the number of experiments as $\delta t ^* \sim N^{-\nicefrac{1}{4}}$. 

The $N^{-\nicefrac{1}{4}}$ scaling can be undertstood intuitively as follows. To achieve a reconstruction which is twice more accurate, the total estimation error must be reduced by a factor of $2$. This requires measuring at a time shorter by a factor of $2$; however, to reduce the statistical error $\norm{\delta \vec{b}_{stat}} \sim \frac{1}{\sqrt{N} \delta t}$ by a factor of $2$ at this twice shorter measurement time, the measurements must be $4$ times more accurate, which requires a number of shots  $16$ times larger.

To estimate the optimal $\delta t^*$ more accurately, we analyze in \App{app:derivative_estimation} the contributions of the statistical and systematic errors to the reconstruction error,
\begin{equation}
     \norm{\delta \vec{c}} = \norm{K^+ \left(\delta \vec{b}_{stat} + \delta \vec{b}_{sys} \right)}.
 \end{equation}
Since the statistical error averages to zero, we find that these contributions factorize, and to leading order in $\delta t$,
\begin{align} \label{error_estimate}
\begin{split} 
    \mathbb{E} \norm{\delta \vec{c}}^2 &= \mathbb{E} \norm{K^+ \delta \vec{b}_{stat} }^2 + \norm{K^+ \delta \vec{b}_{sys} }^2 \\
    &\approx \frac{a_{stat}}{N (\delta t)^2} + a_{sys} (\delta t)^2,
\end{split}
\end{align}
for suitable constants $a_{stat}$ and $a_{sys}$. The optimal time $\delta t^*$, and the corresponding optimal reconstruction error $\Delta^*$ are therefore approximated by
\begin{equation} \label{optimal-dt-first-order}
    \delta t^* \approx \left( \frac{a_{stat}}{a_{sys} N}\right)^{\frac{1}{4}}, \hspace{0.5cm} \Delta^* \approx \left(4 \frac{a_{stat} a_{sys}}{N}\right)^{\frac{1}{4}}/\norm{\vec{c}}.
\end{equation}

The statistical component $a_{stat}$ determines how finite sampling errors degrade the reconstruction. We find that it is given by the spectrum of $\frac{1}{\sqrt{\Nspam}} K$, which quantifies the sensitivity of the SPAM settings to all the Hamiltonian parameters, averaged over the $\Nspam$ SPAM settings,
\begin{equation} \label{stat_error}
    a_{stat} =  \Tr \left(\left( \frac{1}{\Nspam} K^T K\right)^{-1} \right).
\end{equation}

% \begin{equation}
%     \frac{1}{\Nspam}K^T K\underset{p\rightarrow\infty}{\rightarrow}4 \sum _{\beta,[S_j,A_\beta]\neq 0} 
% \delta_{j,j'} \cdot(\frac{1}{3}) ^{\ell_{j,\beta}}, \end{equation}
%  \enote{this equation }
%  where the sum runs over the estimated observables $A_\beta$ that do not commute with the  Hamiltonian term $S_j$, and 
%  $\ell_{j,\beta}$ denotes the size of the support of $[S_j,A_\beta]$ (appendix ?). For a fixed geometry, this depends only on the locality of the 
 
 %which shows that in that limit $K^+$ exists and scales as $\frac{1}{\sqrt{\Nspam}}$ as far as for each basis element $S_j$ there is a constraint which do not commute with it. To see the diagonal shape of $K^TK$, consider the following simple setup of constraints and states: all 1-local constraints and all 1-local repeating product states, and a 2-local Hamiltonian that way all the off diagonal terms of $K^TK$ must vanish and the diagonal terms\anote{I dont see an easy way of explaining the result here without getting into too much details}  \anote{should continue}.

The systematic compenent of the error stems from the higher-order derivatives of the estimated observables. When the systematic error is calculated to leading order in $\delta t$, it is given by 
\begin{equation} \label{sys_error}
    a_{sys} = \norm{K^+ \partial_{tt} \ev{\vec{A}}|_{t=0}}^2,
\end{equation}
where $\partial_{tt} \ev{A} = -\ev{\com{H, \com{H, A}}}$. A more accurate estimate for long times (small $N$) can be calculated using higher orders derivatives (e.g. using $\partial_{ttt} \ev{A}$ as well, see \App{app:systematic-error-est}). Since the systematic error depends on the unknown Hamiltonian we wish to reconstruct, it must be estimated separately; for instance, by computing it with respect to a guess of the true Hamiltonian.

Remarkably, a rough guess of the true Hamiltonian provides an excellent estimate for the reconstruction error and optimal measurement time. We estimated these quantities for the experiments presented in \Fig{fig1}, using the second-order estimate for the systematic error $a_{sys}$ described in \App{app:systematic-error-est}. We assumed that the true Hamiltonian is known to an accuracy of $10\%$, and estimated $a_{sys}$ according to a perturbed version of the true Hamiltonian 
\begin{equation} \label{guess}
    H+H'=\sum_j c_j(1+ e_j)S_j,
\end{equation} where $e_j\sim N(0,{0.1}^2)$. Plugged in \Eq{optimal-dt-first-order}, this estimate reproduced the reconstruction error in \Fig{fig1} near its optimum to an excellent accuracy (\Fig{fig1}, left, gray line). The predicted optimal measurement time matched the empirical optimum even at moderate measurement budgets $N$. While it is tempting to improve the reconstruction using higher order finite difference methods \cite{Dumitrescu2019HamiltonianSystems}, our experiments found an advantage for these methods only at impractical measurement budgets $N \geq 10^{11}$ (see \App{app:higher_order}).

%\enote{Cutbacks?} It seems tempting to improve the reconstruction using higher order finite difference methods, whose approximation error scales as a higher power of $\delta t$ \cite{Dumitrescu2019HamiltonianSystems}. For instance, the error of 3 points forward finite differentiation scales as $\sim \delta t ^2$ instead of $\sim \delta t$. Indeed, we find that these higher order methods lead to a better asymptotic scaling of the reconstruction error with the measurement budget $N$. However, our experiments found an advantage for these higher-order methods only at impractical measurement budgets $N \geq 10^{12}$ (see \App{app:higher_order}).

\begin{figure}  \includegraphics[width=8 cm]{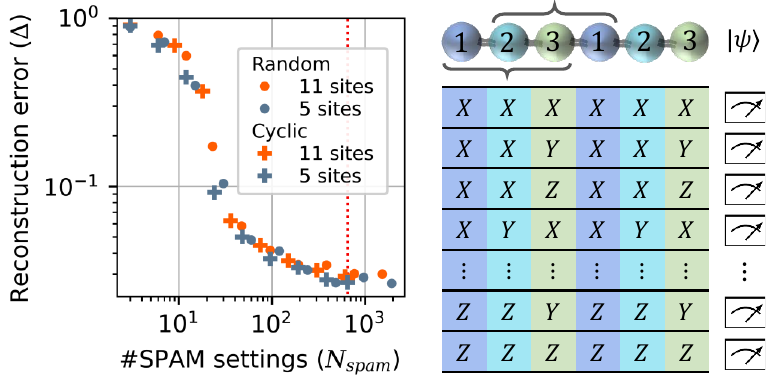} 
  \caption{Left: Reconstruction error as a function of the number of different state preparation and measurement (SPAM) settings for various system sizes and a fixed total number of shots ($N=8\cdot 10^6$). The reconstruction accuracy of both random SPAM settings (circles) and cyclic SPAM settings (crosses) does not depend on system size, and saturates with sufficiently many settings. The saturation point is predicted by the number of settings $\Nspam^*=6^3\cdot 3$ in the full set of cyclic settings (vertical dotted line); for $\Nspam<\Nspam^*$, the crosses correspond to random subsets of this set. Right: the measurements of the cyclic protocol (Overlapping Local Tomography) for the case of $k_\beta = 3$-ranged observables. Partition the lattice to unit cells of length $k_\beta$, then measure each unit cell in an informationally complete set of bases ($3^{k_\beta}$ in one dimension) to obtain its reduced density matrix, repeating the bases periodically across unit cells. The reduced density matrices of overlapping partitions (gray braces) are obtained for free, since the measurement bases of the different partitions are equal up to a permutation. Initial states in the cyclic SPAM scheme are chosen similarly.} 
 \label{fig2}
\end{figure}

\paragraph*{System size scaling.} So far, we numerically estimated the reconstruction error at a fixed number of SPAM settings $\Nspam=648$ and system size $n=7$. How many SPAM settings $\Nspam$ are required for reconstruction, and how does the measurement budget $N$ required for reconstruction scale with system size?

%An important question that arises from examining the protocol we have just described is, how is the system size scaling of it? (both  the experiment settings needed for reconstruction and the reconstruction error itself).
We repeated the reconstruction experiment of random $2$-ranged Hamiltonians [\Eq{H_rand}] with varying numbers of different SPAM settings $\Nspam$ and a constant total budget of measurements $N$. The reconstruction error saturated at $\Nspam^* \approx 10^3$ SPAM settings (\Fig{fig2}, left, dots). Crucially, both $\Nspam^*$ and the reconstruction error remained unchanged when we increased the system size from $n=5$ to $n=11$, keeping the total measurement budget fixed. The larger system yielded the same reconstruction error, saturating at the same number of SPAM settings.

To understand the saturation of the reconstruction error as a function of the number of SPAM settings $\Nspam$, we consider the $\Nspam\rightarrow \infty$ limit of all SPAM settings. Averaging over all product states and measurement bases, the correlation matrix $\frac{1}{\Nspam} K^T K$ converges to a diagonal matrix with system size independent entries (see \App{app:analytic_K}). The entry for each Hamiltonian term depends only on its support and on the set of measured observables. For example, if only single-site observables $A_\beta$ are used,
\begin{equation} \label{analytic_spectrum}
      \frac{1}{\Nspam} \left( K^T K \right)_{j,j'} \underset{p\rightarrow\infty}{\rightarrow} 4\cdot\frac{2k_j}{3^{k_j + 1}} \delta_{j,j'},
\end{equation}
%\anote{$1/p K^TK=1/p \sum_\beta \sum_\alpha \comm{S_i}{A_\beta}_\alpha \comm{S_j}{A_\beta}_\alpha \rightarrow \frac{8}{3}k_j\delta_{ij}$ When neglecting boundary terms} \enote{We need to have this derivation in Appendix C so we can be sure about this, and interested readers can also see how this result depends on the locality of the constraints etc. I see where I probably missed a 'k' factor, but it should also decrease exponentially with 'k' due to the commutators becoming less local $\rightarrow$ smaller expectation value in a product state, right?}\anote{right, maybe what I wrote is more relevant in the cyclic scheme where we take into account the exponential factor we get from the repeating measurements per locality}
where $k_j$ is the weight of the corresponding Hamiltonian term $S_j$, which is the size of its support. 

%Importantly, the set of equations converges to this well-conditioned form at a system-size independent number of SPAM settings $\Nspam$. The $\Nspam\rightarrow \infty$ limit does not actually require to average over all possible SPAM settings; it suffices to uniformly sample all local configurations within the shared support of every pair of overlapping Hamiltonian terms $S_i$, $S_j$.

The insensitivity of the reconstruction to system size is apparent in the $\Nspam\rightarrow \infty$ average over SPAM settings [\Eq{analytic_spectrum}]. The system-size independence of the correlation matrix entries implies a system-size independent statistical error [\Eq{stat_error}] per Hamiltonian term, as the average singular value of $K$ becomes constant. The systematic error [\Eq{sys_error}] per Hamiltonian term  is also constant for short-ranged Hamiltonians, since each Hamiltonian term affects a fixed number of observables. 

More generally, the system-size independence follows from the local structure of the reconstruction algorithm, since the equations for each Hamiltonian term involve only observables in its vicinity \cite{Bairey2019}. This local structure also implies that the $\Nspam \rightarrow \infty$ limit of \Eq{analytic_spectrum} does not require all SPAM settings; only the local configuration of each initial state and measurement basis matters. To predict the number of SPAM configurations required to saturate the reconstruction error, we now devise a structured SPAM scheme. This structured scheme uses a finite number of 'cyclic' SPAM configurations which are equivalent to the $\Nspam \rightarrow \infty$ average over all configurations. 

The cyclic measurement bases of the structured scheme are chosen to estimate all observables of a constant range $k_\beta$ in a given initial state. We construct these measurement by partitioning the lattice to unit cells of length $k_\beta$, and iterating over all measurement configurations within a unit cell (\Fig{fig2}, right). The configurations are copied between unit cells, defining a set of $3^{k_\beta}$ measurement bases in one dimension, or $3^{k_\beta^D}$ bases in a $D$-dimensional lattice, independent of system size. Inspired by \iRef{Cotler2020QuantumTomography}, we call this measurement process 'Overlapping local tomography', since it characterizes all the $k_\beta$-ranged reduced density matrices of a given state.

The cyclic initial states are constructed similarly to the cyclic measurement bases. For the purposes of reconstructing a $k$-ranged Hamiltonian, cyclic initial states with periodicity $2k - 1$ are equivalent to the full set of initial product states. This periodicity is chosen to uniformly cover the shared support of every pair of overlapping Hamiltonian terms $S_i$, $S_j$ (see \App{app:analytic_K}). For the $k=2$-ranged Hamiltonians and $1$-ranged observables considered so far, this structured scheme consists of $\Nspam^* = 648=6^3 \cdot 3$ configurations, corresponding to the saturation point of \Fig{fig2}. These structured configurations yield a similar reconstruction accuracy as the random configurations (\Fig{fig2}, left, crosses).%The saturation of the reconstruction error with the number of SPAM settings $\Nspam$ is explained by the $\Nspam\rightarrow \infty$ limit of $\frac{1}{\Nspam} K^T K$. %\rep{The spectrum of $\frac{1}{\sqrt{\Nspam}} K$ can be calculated numerically for a given choice of experimental settings. Moreover, this spectrum can be calculated analytically in the limit $\Nspam\rightarrow \infty$, which averages over all experimental settings shown in \Fig{fig1}.}{} 

\begin{figure}  \includegraphics[width=8 cm]{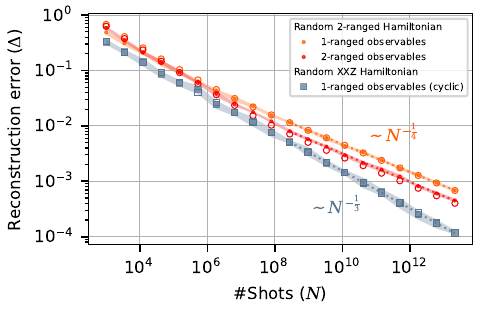} 
  \protect
  \caption{Reconstruction error as a function of the total number of shots. Warm colors: reconstruction of random $2$-ranged Hamiltonians using all observables up to range $k_\beta = 1$ (orange) and $k_\beta = 2$ (red). Addition of $2$-ranged observables improves the reconstruction accuracy at the same measurement budget $N$. Blue: reconstruction of random XXZ Hamiltonians using $1$-ranged observables in a cyclic measurement scheme with periodicity $2$. Symmetries of the XXZ model lead to an easier reconstruction, with error scaling asymptotically as $N^{-1/3}$ in contrast to the $N^{-1/4}$ scaling predicted for generic Hamiltonians.} 
 \label{fig3}
\end{figure}

\paragraph*{Choosing the observables.}
%\anote{find a connection sentence} \enote{Maybe motivate this section by reminding that the same set of measurements can be used to estimate various observables, but the estimation accuracy decreases with locality. So, do these higher-order correlators aid in the reconstruction? +5 points if it naturally connects with the previous section :P}\anote{ \rep{I now realized that the estimation accuracy does not change with locality in the random scheme as we presented it [depends if we do multiple random measurements for each state], the graphs in figs 4,2 on the other hand where done with the option of all ps are measured with all measurement bases }{sorry, I was wrong only in the cyclic scheme we did it}} 
An important factor in the reconstruction is the choice of observables ${\{
A_\beta\}}$ used for building the linear equation \Eq{linear_reconstruction}. Each measurement setting can be used to estimate $2^n-1$ commuting Pauli observables. Our previous simulations used only single-site observables in the reconstruction. How does the choice of estimated observables affect the reconstruction quality?

To test this, we repeated our simulations with different choices of observables. We used a fixed dataset of experimental outcomes to estimate all observables up to range $k_{\beta}$ for $k_{\beta}=1,2,3$. While reconstruction with $k_{\beta}=1$ observables was improved by the addition of $k_{\beta}=2$ observables (\Fig{fig3}, warm colors), the addition of $k_{\beta}=3$ degraded the reconstruction (not shown). Correlations between estimated observables cause the degradation for $k_\beta = 3$, as well as the deviation between predicted and measured errors for $k_\beta=2$ (see \App{app:opt_time}); we conjecture that observables up to a range matching that of the Hamiltonian $k_{\beta}=k$ are optimal more generally.

%Motivated by the previous section about the system size independent we can dividing the possible observables to groups based on their spatial range and ask whether we can use only local observables? and how does the locality of the observables depends on the locality of the Hamiltonian? 
%So far, in our simulations we have used $2-local$ random Hamiltonians  where we chose to estimate all $2-local$ observables for the reconstruction. 

\paragraph*{Symmetry eases Hamiltonian reconstruction.}
In contrast to the Hamiltonians we simulated so far, realistic Hamiltonians are not fully random and often have symmetries. How does symmetry affect the reconstruction quality and the accuracy of the our predictions? 

To address this question, we repeated our experiments with random XXZ Hamiltonians:
\begin{equation} \label{XXZ_hamil}
  H_{XXZ}=\sum_{i=0}^{n-1}c_i Z_i Z_{i+1}+a_i(X_i X_{i+1}+Y_i Y_{i+1}),
\end{equation}
where $c_i,a_i\sim N(0,1)$ are drawn from a Gaussian distribution. These $XXZ$ Hamiltonians are invariant under three anti-unitary symmetries $\forall i: \sigma_i^\alpha \mapsto -\sigma_i^\alpha$,  each of which flips one of the Bloch axes $\alpha$ for all spins $i$ \cite{meng-cheng}.
Interestingly, the reconstruction error of these Hamiltonians in a cyclic measurement scheme decreased more rapidly with the total number of shots $\Delta \sim N^{-1/3}$ compared to the naive prediction $N^{-1/4}$  of \Eq{optimal-dt-first-order} (\Fig{fig3}, blue). We find that for initial states that respect at least one of the anti-unitary symmetries of the Hamiltonian, the systematic error $a_{sys}$ [\Eq{sys_error}] vanishes to first order (see \App{app:xxz}). As a result, the systematic error scales as $\norm{\delta \vec{b}_{sys}}\sim \order{\delta t^2}$ instead of $\order{dt}$, explaining the $N^{-1/3}$ scaling of the reconstruction error.

\paragraph*{Conclusions.} We have shown that a system-size independent number of short-time measurements fully characterizes short-ranged Hamiltonians to any given accuracy. Fixing the total number of experimental shots, the error decreases with the number of experimental configurations up to a saturation point. While the error is very sensitive to the measurement time, the optimal time can be predicted using minimal prior knowledge.

Short-time measurements for Hamiltonian reconstruction benefit from computational and analytical tractability, as well as favorable system-size scaling. However, the reconstruction scales poorly with the number of shots, since the derivative estimation requires evolution times which decrease with the desired reconstruction accuracy. Longer measurement times can be utilized by substituting the differential Ehrenfest equation [\Eq{Ehrenfest}] by an integral equation for the dynamics of observables over a time period. Altenatively, a higher-order Taylor expansion of the observable dynamics would lead to a non-linear set of equations, allowing longer measurement times at the cost of higher computational effort as in the recent \iRef{Haah2021OptimalStates}. The optimal measurement time for these higher order methods can be estimated using extensions of the method we presented here.

Alternatively, recent works discuss a scheme for dynamical Hamiltonian learning which is based on on energy conservation \cite{Bentsen2019IntegrableCavity, Li2019}. These works look for a local observable which does not change with time. In the absence of other local symmetries, the only conserved observable is the Hamiltonian. Thus, each initial state or measurement time gives a single linear equation for the Hamiltonian, such that $\order{n}$ suffice to uniquely identify it. The equations remain computationally tractable for long evolution times, increasing their sensitivity at a fixed system size. In our short-time approach, each measurement settings yields not one but $\order{n}$ local equations for the Hamiltonian, providing an advantage at large system sizes.

%A promising approach learns Hamiltonians from the short-time dynamics of product states \cite{Shabani2011, DaSilva2011}. At longer times, the dynamics of generic Hamiltonians become exponentially complex, posing challenges for methods based on measurements at various times \cite{Han2021TomographyLearning, Sone2017, Zhang2014, Zhang2015, DiFranco2009, Burgarth2009}. To deal with this exponential complexity, some approaches assume partial control of the quantum system to be characterized \cite{Wang2015, Valenti2019, Krastanov2019StochasticVariables, Valenti2021ScalableDynamics}, or employ an additional trusted quantum simulator \cite{Granade2012, Wiebe2014, Wiebe2014a, Wiebe2015, Wang2017}. In contrast, short-time dynamics are linear in the Hamiltonian parameters, allowing a computationally-efficient reconstruction \cite{DaSilva2011}. Product states can be rapidly prepared experimentally, providing an advantage over approaches which learn the dynamics from their steady states \cite{Qi2017, Chertkov2018, Greiter2018, Rudinger2015, Kieferova2017, Kappen2018, Bairey2019, Bairey2019a, Dumitrescu2019HamiltonianSystems, Evans2019ScalableLearning, Anshu2020}. By preparing various product states, and measuring the time-derivative of different observables, any short-ranged Hamiltonian can be learned in a scalable fashion. 

\begin{acknowledgments}
We thank Itai Arad, Barak Gur and Tasneem Watad for useful discussions.  We acknowledge financial support from the
 European Research Council (ERC) under the European Union Horizon 2020 Research and Innovation Programme (Grant Agreement No. 639172), the Defense Advanced Research Projects Agency through the DRINQS program, grant No. D18AC00025, and from the Israel Science Foundation within the ISF-Quantum program (Grant No. 2074/19). The content of the information presented here does not necessarily reflect the position or the policy of the U.S. government, and no official endorsement should be inferred. 
\end{acknowledgments}

\bibliography{references}

%% LyX 2.1.2 created this file.  For more info, see http://www.lyx.org/.
%% Do not edit unless you really know what you are doing.

%\renewcommand{\theequation}{S\arabic{equation}}
\renewcommand{\thefigure}{S\arabic{figure}}
\appendix

\section{Derivative estimation}\label{app:derivative_estimation}

\subsection{Full derivation of the optimal measurement time} \label{app:opt_time}
\subsubsection{Full derivation of $\mathbb{E} \norm{\delta\vec{c}}$ bound} 

We want to estimate how $\delta t$ affects $\mathbb{E}\norm{\delta\vec{c}}$. The effect of $\delta t$ on $\delta \vec{b}$ can be written as sum of two contributions $\delta\vec{b}=\delta\vec{b}_{sys}+\delta\vec{b}_{stat}$ where
\begin{equation} \label{b_sys}
    \delta b_{sys_{\left(\alpha, \beta\right)}}=b_{\left(\alpha, \beta\right)} -\frac{\ev{A_\beta }_{\psi_\alpha (\delta t)}-\ev{A_\beta}_{\psi_\alpha (0)}}{\delta t},\
\end{equation}
with expectation values evaluated exactly (in the limit of infinite shots), is the systematic error of the estimation (\Eq{finite_difference}) and 
\begin{equation} \label{b_stat}
\delta b_{stat_{\left(\alpha, \beta\right)}} = \frac{\delta \ev{A_\beta }_{\psi_\alpha (\delta t)}}{\delta t}\sim G\left(0, \frac{1}{{\delta t}^2}\frac{\Nspam}{N}\right),
\end{equation}
with expectation values estimated using $\frac{N}{\Nspam}$ shots, is the statistical error due to finite sampling. Writing $\delta \vec{c} = K^+\delta \vec{b} = K^+\left( \delta \vec{b}_{stat} +\delta \vec{b}_{sys} \right) $ as sum of two independent errors allows us to find a value for $\mathbb{E}\norm{\delta\vec{c}}^2$ using the following lemma on the expectation value of a quadratic form.

\begin{lemma}
Let $\vec{x}$ be a random vector with mean $\vec{\mu}$ and covariance matrix $\Sigma$ then for any $A\in M_{n,m}(\mathbb{R})$,  $\mathbb{E}\norm{A\vec{x}}^2 = \Tr\left(\Sigma A^TA\right)+\norm{A\vec{\mu}}^2$.
\end{lemma}

\begin{proof}

Note that 
\[\norm{A\vec{x}}^2=\vec{x}^TA^TA\vec{x}\]
Because
$\norm{A\vec{x}}$ is a scalar, we can write \[\norm{A\vec{x}}^2=\Tr\left(\vec{x}^TA^TA\vec{x}\right)\]
Using the cyclic property of the trace \[\norm{A\vec{x}}^2=\Tr\left(A^TA\vec{x}\vec{x}^T\right)\]

As trace and matrix multiplication are essentially additions and scalar multiplications we can write \[\mathbb{E}\Tr\left(A^TA\vec{x}\vec{x}^T\right)=\Tr\left(A^TA\mathbb{E}\left(\vec{x}\vec{x}^T\right)\right)\]
Using the definition of the covariance matrix \[\mathbb{E}\norm{A\vec{x}}^2 = \Tr\left(A^TA\left(\Sigma+\vec{\mu}\vec{\mu}^T\right)\right)\] and from linearity of the trace \[\mathbb{E}\norm{A\vec{x}}^2 = Tr\left(A^TA\Sigma\right)+\Tr\left(A^TA\vec{\mu}\vec{\mu}^T\right)\] Again from the cyclic properties of the trace and the fact that $\Tr\left(\vec{\mu}^TA^TA\vec{\mu}\right)=\norm{A\vec{\mu}}^2$ (as it is a scalar):
\[\mathbb{E}\norm{A\vec{x}}^2 = \Tr\left(\Sigma A^TA\right)+\norm{A\vec{\mu}}^2\]

\end{proof}

In our case $\delta\vec{b}$ is a random vector with mean $\delta\vec{b}_{sys}$ and covariance matrix $\Sigma_{\delta \vec{b}_{stat}}= \mathbb{E}\delta \vec{b}_{stat}\delta \vec{b}_{stat}^T$.
Using that lemma for $\mathbb{E}\norm{\delta \vec{c}}^2=\mathbb{E}\norm{K^+\delta \vec{b}}^2$ we can find a bound on $\mathbb{E}\norm{\delta \vec{c}}$ using Jensen inequality $\mathbb{E}\norm{\delta \vec{c}}\leq\sqrt{\mathbb{E}\norm{\delta \vec{c}}^2}$:
\[\mathbb{E}\norm{\delta \vec{c}}\leq\sqrt{\Tr\left(\mathbb{E}\left(\delta\vec{b}_{stat}{\delta\vec{b}_{stat}}^T\right) {K^+}^TK^+\right)+\norm{K^+\delta\vec{b}_{sys}}^2}.\]
Assuming each entry of $\delta\vec{b}_{stat}$ is independent random variable with variance $\sigma_{(\alpha,\beta)}^2$ and zero mean we get that
\begin{equation} \label{delta-c-bound}
    \mathbb{E}\norm{\delta \vec{c}}\leq \sqrt{\Tr\left(\Sigma_{\delta \vec{b}_{stat}} {K^+}^TK^+\right)+\norm{K^+\delta\vec{b}_{sys}}^2},
\end{equation} where $\Sigma_{\delta \vec{b}_{stat}}$ is diagonal with elements $\sigma_{(\alpha,\beta)}^2$.
Assuming $\sigma_{(\alpha,\beta)}\sim \left(\frac{N}{\Nspam}\right)^{-\nicefrac{1}{2}}\frac{1}{\delta t}$ we can define \[{a}_{stat}=\Tr\left(\Sigma_{\delta \vec{b}_{stat}}^* \Nspam {K^+}^TK^+\right)\] where $\Sigma_{\delta \vec{b}_{stat}}^*=\frac{N}{\Nspam}{\delta t}^2\Sigma_{\delta \vec{b}_{stat}}$ is a diagonal matrix which is independent of $N, p, \delta t$ and \[a_{sys}=\frac{1}{{\delta t}^2}\norm{K^+\delta\vec{b}_{sys}}^2\]
Assuming $\delta\vec{b}_{sys}\sim \delta t$ (first order estimation of the systematic error) we get \Eq{error_estimate}, and if we assume that $\Sigma_{\delta \vec{b}_{stat}}^*=\mathbb{1}$ (i.e. all the observables have the same statistical error magnitude) using the cyclic property of the trace we can write ${a}_{stat}=\Tr\left( \Nspam  \left(K^TK\right)^{-1}\right)$.

In practice, the standard deviations varies between observables depending on the initial state. Moreover, when different observables are estimated from the same set of measurements, off-diagonal entries appear in the covariance matrix. These correlations explain the deviation between the predicted and observed accuracies in \Fig{fig3}. They can be taken into account in order to improve the prediction of the reconstruction, and to predict the optimal set of estimated observables.

\subsubsection{Estimating the optimal $\delta t$}

In \Eq{error_estimate} and subsequent derivations, we have assumed that the statistical error scales as $\delta\vec{b}_{stat}\propto {\left(\frac{N}{\Nspam}\right)}^{-\nicefrac{1}{2}}\delta t^{-1}$ and the systematic error as $\delta\vec{b}_{sys}\propto \delta t$. We can examine a more general case with  \begin{gather*} 
\delta \vec{b}_{stat}=\vec{a}_{stat}{\left(\frac{N}{\Nspam}\right)}^{-\nicefrac{1}{2}}{\delta t}^{-\gamma} \\
\delta \vec{b}_{sys}=\vec{a}_{sys}{\delta t}^{\lambda}+\order{\delta t^{\lambda+1}}
\end{gather*}
for $\lambda, \gamma \ne 1$. $\vec{a}_{stat},\vec{a}_{sys}$ are defined such that they are independent of $\delta t$.
The optimal $\delta t$ is achieved when $\mathbb{E}\norm{\delta \vec{c}}$ is minimal. As an approximation, we use the bound $\norm{\delta c}_{bound}=\sqrt{\Tr\left(\Sigma_{\delta \vec{b}_{stat}} {K^+}^TK^+\right)+\norm{K^+\delta\vec{b}_{sys}}^2}$ of \Eq{delta-c-bound}, 
and minimize it with respect to $\delta t$. To leading order in $\delta t$ we get
\begin{multline*}
\norm{\delta c}_{bound}^2=\frac{{\delta t}^{-2\gamma}}{N} \Tr\left(\Sigma_{\vec{a}_{stat}}{K^+}^TK^+\right)\\
+{\delta t}^{2\lambda}\norm{K^+\delta\vec{a}_{sys}}^2.
\end{multline*}
The derivative is given by 
\begin{multline*}
\dfrac{d}{d(\delta t)}\norm{\delta c}_{bound}^2=-2\frac{\gamma {\delta t}^{-2\gamma-1}}{N} \Tr\left(\Sigma_{\vec{a}_{stat}}{K^+}^TK^+\right)\\
+2\lambda{\delta t}^{2\lambda-1}\norm{K^+\delta\vec{a}_{sys}}^2
\end{multline*}
The minimum is achieved when the derivative vanishes, which occurs at
\begin{equation} \label{optimal-dt-generalized}
\delta t^* = \left(\frac{1}{N}\cdot\frac{\gamma\Tr\left(\Sigma_{\vec{a}_{stat}}{K^+}^TK^+\right)}{\lambda\norm{K^+\delta\vec{a}_{sys}}^2}\right)^{\nicefrac{1}{2\left(\gamma+\lambda\right)}},
\end{equation}
which corresponds to reconstruction error
\begin{multline}\label{optimal-rec-generalized}
\norm{\delta c^*}_{bound}=\left(\sqrt{\frac{1}{1-\kappa}}\frac{1}{\sqrt{N}}\sqrt{\Tr\left(\Sigma_{\vec{a}_{stat}}{K^+}^TK^+\right)}\right)^{1-\kappa} \\
\cdot\left(\sqrt{\frac{1}{\kappa}}\norm{K^+\delta\vec{a}_{sys}}\right)^\kappa,
\end{multline}
where $\kappa = \frac{\gamma}{\lambda+\gamma}$. This result generalizes \Eq{optimal-dt-first-order}, where $\gamma=\lambda=1$ and $\kappa=\nicefrac{1}{2}$.

\subsection{Estimation used in the Hamiltonian reconstruction simulations}\label{app:systematic-error-est}
In our simulations, the systematic error due to the finite derivative method was estimated using a perturbed version of the real Hamiltonian, used to calculate the higher derivatives $\partial_{tt} \ev{A_j}_{t=0}$ and $\partial_{ttt} \ev{A_j}_{t=0}$ of the observables.
The systematic error is defined as
\[
a_{sys} = \norm{K^+ \left( \partial_{t}\ev{\vec{A}}|_{t=0}-\frac{\ev{\vec{A}}_{\delta t}-\ev{\vec{A}}_{0}}{\delta t}  \right)}^2
.\]
Writing $\ev{\vec{A}}_{\delta t}$ as a Taylor expansion around $t=0$ gives:
\begin{multline}
\ev{\vec{A}}_{\delta t}=\ev{\vec{A}}_{0}+\delta t \partial_{t}\ev{\vec{A}}|_0+\frac{{\delta t}^2}{2} \partial_{tt}\ev{\vec{A}}|_0\\
+\frac{{\delta t}^3}{6} \partial_{ttt}\ev{\vec{A}}|_0+\order{{\delta t}^4}
\end{multline}
Then for first order in $\delta t$ the systematic error is:
\[
a_{sys} = \norm{K^+ \left( \frac{1}{2}\partial_{tt} \ev{\vec{A}}|_{t=0} \right)}^2
\]
and for second order in $\delta t$:
\begin{equation}\label{a-sys-2nd-order}
a_{sys} = \norm{K^+ \left( \frac{1}{2}\partial_{tt} \ev{\vec{A}}|_{t=0} + \frac{\delta t}{6}\partial_{ttt}\ev{\vec{A}}|_{t=0} \right)}^2
\end{equation}

Evaluating this second order estimate for $a_{sys}$ requires to know $\delta t$ which is our ultimate goal to find. In order to bypass this difficulty, we used the first order estimate of $a_{sys}$ in \Eq{optimal-dt-first-order} to get a first order estimate of $\delta t$, and used that estimate in the second order estimate of $a_{sys}$ to get a better prediction of $\delta t$.

\begin{figure}  \includegraphics[width=8.6cm]{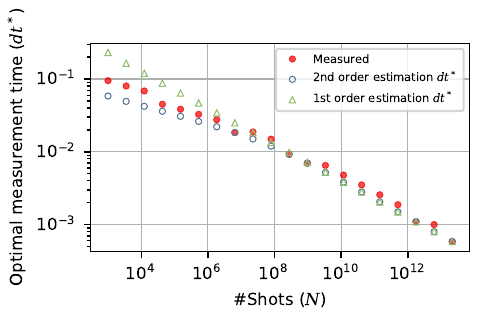} 
  \protect
  \caption{Predicting the optimal measurement time using first order (green hollow triangles, \Eq{sys_error}) or second order (blue circles, \Eq{a-sys-2nd-order}) methods in a random 2-local Hamiltonian compared to the measured optimal time (red dots). Here we used $2$-ranged observables in a randomized measurement scheme.}
 \label{fig7}
\end{figure}

\subsection{Higher order finite difference approximations} \label{app:higher_order}
Examining the analytical results we got for the reconstruction error bound in \Eq{optimal-rec-generalized} we can try to improve the reconstruction using higher-order estimates for the derivative. Instead of using the finite difference method with a single forward point (\Eq{finite_difference}), we can use forward finite difference method with $\lambda>1$ points
\begin{equation} \label{finite_difference_general}
    \partial_t \ev{A}_{\psi}=\frac{\sum_{r=0}^{\lambda}w_{r}\ev{A }_{\psi (r\cdot \delta t)}}{\delta t}+\order{\delta t^\lambda}
\end{equation}
where $w_r$ are the weights of each point such that the estimation error is $\order{\delta t^\lambda}$.
In that case we will achieve better asymptotic behaviour with respect to the number of measurements (as $\norm{\delta c^*}_{bound}\propto \sqrt{N}^{\frac{\lambda}{\lambda+1}}$, \Eq{optimal-rec-generalized}) but the measurement budget will split between $\lambda$ measurement times. The measurement budget at different times can allocated wisely by finding a set $\{N_r\}_{r=1}^{\lambda};\sum_{r=1}^\lambda N_r=N$ which minimizes the variance of the statistical error $\sigma^2=\sum_{r=1}^{\lambda} {w_r}^2/{N_r}$, which is achieved for
\begin{equation}
N_r=N\frac{\abs{w_r}}{\sum_{r=1}^{\lambda}\abs{w_r}}.
\end{equation}
Hence, the $N$ used in the statistical error should be replaced with $N_{effective}=N\left(\frac{1}{\sum_{r=1}^{\lambda}\abs{w_r}}\right)^2$, 
which yields an additional $\left(\sum_{r=1}^{\lambda}\abs{w_r}\right)^{\frac{\lambda}{\lambda+1}}$ factor to the reconstruction error bound because of the statistical error (\Eq{optimal-rec-generalized}). Additional factors can arise due to the systematic error.

We examined the results of reconstruction using finite difference with $\lambda=2$ uniformly spaced forward points
\begin{equation} \label{finite_difference_3_points}
    \partial_t \ev{A}_{\psi}=\frac{-\frac{3}{2}\ev{A}_{\psi(0)}+2\ev{A}_{\psi(\delta t)}-\frac{1}{2}\ev{A}_{\psi(2\delta t)}}{\delta t}+\order{\delta t^2}.
\end{equation}
Comparing the results we got using 1 forward point finite difference (\Eq{finite_difference}) to the results using 2 forward points finite difference (\Eq{finite_difference_3_points}) we observed that the scaling advantage from the the 2 forward points estimation method yields better reconstruction only for $N>10^{11}$ as can be seen in \Fig{fig6}. 
\begin{figure}  \includegraphics[width=8.6cm]{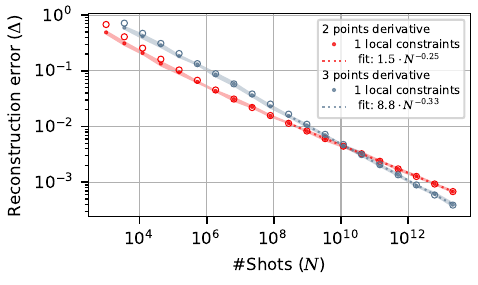} 
  \protect
  \caption{Reconstruction error of a random 2-local Hamiltonian as a function of the total number of shots using different methods of derivative estimation: 2 points forward finite difference (blue dots, \Eq{finite_difference}) and 3 points forward finite difference (red dots, \Eq{finite_difference_3_points}).}
 \label{fig6}
\end{figure}

\section{Analytic estimate of the statistical reconstruction error} \label{app:analytic_K}

A unique reconstruction requires a set of initial states and measurement bases leading to a full rank reconstruction matrix $K$. More generally, the choice of experiments determines the spectrum of $K^T K$, which quantifies the sensitivity of the experiments to the Hamiltonian terms through the statistical error of \Eq{stat_error}. We now derive the analytical form of $K^T K$ when averaged over all experimental configurations, and show that it is sufficient to average over a finite set of local configurations.

For a given set of initial states $\ket{\psi_\alpha}$ and observables $A_\beta$,
\begin{equation} \label{all_p_KTK}
(K^T K)_{ij}=-\sum_{\alpha,\beta}\ev{i[S_i,A_\beta ] }_{\psi_\alpha}\ev{i[S_j,A_\beta ] }_{\psi_\alpha},
\end{equation}
by the definition of $K$ [\Eq{K_element}]. Here, $A_\beta$ denote all the observables compatible with the measurement basis for the initial state $\ket{\psi_\alpha}$ up to the chosen locality. When averaging over all possible configurations, each of the $6^n$ initial states is measured in $3^n$ measurement bases. However, for a given initial state $\ket{\psi_\alpha}$, each observable $A_\beta$ is compatible with many measurement bases. An observable with support $s \subseteq [1,\dots,n]$ is compatible with $\frac{1}{3^{\abs{s}}}$ of all measurement bases; for instance, $X_1$ is compatible with $\frac{1}{3}$ of all bases. Therefore, the average $\frac{1}{\Nspam} (K^T K)_{ij}$  over all configurations $\Nspam=\frac{1}{3^n \cdot 6^n}$ takes the form
\begin{equation} \label{average_gram}
\sum_{\beta} \frac{1}{3^{\abs{s_\beta}}} \sum_{\alpha}  \frac{1}{6^n} \ev{i[S_i,A_\beta ] }_{\psi_\alpha}\ev{i[S_j,A_\beta ] }_{\psi_\alpha}.
\end{equation}
where $\beta$ runs over all estimated observables. We will now show that this average
\begin{enumerate}
    \item Vanishes for off-diagonal elements where $S_i \neq S_j$.
    \item For the diagonal element corresponding to $S_i$, it counts the number of observables $A_\beta$ that fail to commute with $S_i$, normalized by the weight $s_\beta$ of $A_\beta$ and the weight ${s_{i,\beta}}$ of the commutator $\com{S_i, A_\beta}$.
    \item Depends only on initial state and measurement configurations on sites supported by $S_i$ or $S_j$.
    \item Vanishes identically for \emph{any} initial state $\ket{\psi_\alpha}$ when $S_i, S_j$ do not overlap at least on one site, so that averaging over $2k_h - 1$-local configurations suffices.
\end{enumerate}
For the first two claims, we invoke the following formula for product state averages of Pauli strings \cite{Li2019}:
\begin{equation} \label{pauli_average}
    \frac{1}{6^n} \sum_{\alpha} \ev{P_r}_{\psi_\alpha}\ev{P_r'}_{\psi_\alpha} = \frac{1}{3^{\abs{s_r}}} \delta_{r,r'},
\end{equation}
which states that the average over all initial states vanishes whenever the two Pauli strings are different, and decays with their support $\abs{s_r}$ otherwise. In our case, $i \com{S_i, A_\beta} = 2 P_{r_{i,\beta}}$ is a Pauli string since both $S_i$ and $A_\beta$ are.

\paragraph*{Off-diagonal terms vanish.} Considering \Eq{pauli_average}, it is sufficient to show that  different Hamiltonian terms $S_i \neq S_j$ have different commutators with a shared Pauli observable $ [S_i,A_\beta ] \neq [S_j,A_\beta ]$, assuming that both commutators are non-zero. To see this, we recall that any pair of Paulis either commutes or anticommutes. Therefore, if $\com{S_i, A_\beta} \neq 0$, we can write $\com{S_i, A_\beta} = \theta_{i,\beta} S_i A_\beta$ for some $\theta_{i, \beta} \neq 0$, and similarly for $S_j$. If $\com{S_i, A_\beta} = \com{S_j, A_\beta}$, then $\theta_{i,\beta} S_i A_\beta = \theta_{j,\beta} S_j A_\beta$, and $S_i = \theta_{i, \beta}^{-1} \theta_{j, \beta} S_j$, in contradiction to our assumption.

\paragraph*{Diagonal elements count non-commuting observables.} For diagonal elements $i=j$, \Eq{pauli_average} gives a positive contribution whenever $\com{S_i, A_\beta} \neq 0$. The average over initial states $\alpha$ yields a contribution $\frac{4}{3^{\abs{s_{i,\beta}}}}$ which depends only on the size of the support $s_{i,\beta}$ of the commutator $\com{S_i, A_\beta}$, such that \Eq{average_gram} becomes
\begin{equation} \label{diagonal_entries}
\sum_{A_\beta: \com{S_i, A_\beta} \neq 0} \frac{4}{3^{\abs{s_\beta} + \abs{s_{i,\beta}}}}, 
\end{equation}
where the sum runs over all observables $A_\beta$ up to the chosen observable locality that fail to commute with $S_i$, and the factor 4 comes from the two commutators. For example, for single-site $S_i$ and $A_\beta$ we get $4\cdot \frac{2}{9}$, since $\frac{2}{3}$ of all single-site Paulis $A_\beta$ fail some commute with a single-site Hamiltonian Pauli $S_i$; and $\frac{1}{3}$ of the Pauli measurement bases match each given commutator.

More generally, since our measurement bases average equally over all observables $A_\beta$ with the same support, the value of \Eq{average_gram} depends only on the support of $S_i$, i.e. its locality and range. It can be calculated explicitly for a given set of estimated observables $A_\beta$.

\paragraph*{Only local configurations matter.} We now show that for each $i,j$, \Eq{average_gram} depends only on local configurations within the support of $S_i$ and $S_j$. Namely, it does not depend on the initial state configuration on sites supported by the observable $A_\beta$ but not by any of the Hamiltonian terms. We assume that all observables $A_\beta$ with a given support are sampled equally, as in the overlapping local tomography ('cyclic') measurement scheme.

Consider a fixed initial state $\ket{\psi_\alpha}$, and the average over all observables $A_\beta$ with a given support $s$:
\begin{equation} \label{local_average}
    \frac{1}{3^{\abs{s}}} \sum_{A_\beta: \text{supp}(A_\beta) = s} \ev{i[S_i,A_\beta ] }_{\psi_\alpha}\ev{i[S_j,A_\beta ] }_{\psi_\alpha}.
\end{equation}
The Paulis of $A_\beta$ on the sites are untouched by $S_i$ or $S_j$ factor out of both commutators. When we average over all observables with the same support, only a single observable configuration on those sites matches the initial state $\ket{\psi_\alpha}$, regardless of the initial state or Hamiltonian terms. Formally, we split $A_\beta$ to the part that intersects $S_i$ or $S_j$ and the part that does not,
\begin{equation}
    A_\beta = A_{\beta}^{s \cap s_{ij}} \cdot A_{\beta}^{s \setminus s_{ij}} = \prod_{i=1}^{\abs{s \cap s_{ij}}} \sigma_{m_i}^{\beta_{m_i}} \prod_{i=\abs{s \cap s_{ij}}+1}^{\abs{s}} \sigma_{m_i}^{\beta_{m_i}}
\end{equation}
where $m$ is a site index, $\beta_m \in \lbrace X,Y,Z \rbrace$ is a Pauli index, and $s_{ij} \EqDef \text{supp}(S_i) \cup \text{supp}(S_j)$. We can then write \Eq{local_average} as
\begin{align}
    \frac{1}{3^{\abs{s}}}\sum_{\beta_{m_1},\dots,\beta_{m_{\abs{s \cap s_{ij}}}}} & \ev{i[S_i,A_{\beta}^{s \cap s_{ij}} ] }_{\psi_\alpha}\ev{i[S_j,A_{\beta}^{s \cap s_{ij}} ]}_{\psi_\alpha}  \cdot \\
     & \cdot \sum_{\beta_{\abs{s \cap s_{ij}} + 1},\dots, \beta_{\abs{s}}} \ev{A_{\beta}^{s \setminus s_{ij}}}_{\psi_\alpha}^2.
\end{align}
For any given initial state, the second sum over all observable configurations that are untouched by $S_i$ or $S_j$ gives 1. Since the first sum depends only on the initial state configurations in $s \cap s_{ij}$, it is sufficient to average over those.

\paragraph*{Distant off-diagonal terms vanish for any product state.} We now show that $(K^T K)_{ij}$ vanishes for each of our initial states whenever $S_i, S_j$ do not overlap. For a given Pauli observable $A_\beta$, we done its part that intersects $S_i$ by $A_\beta^{s_i}$; namely, we partition $A_\beta$ to
\begin{equation}
    A_\beta = A_\beta^{s_i} \cdot A_\beta^{s \setminus s_i}.
\end{equation}
If $A_\beta$ commutes with $S_i$, its contribution to \Eq{average_gram} vanishes. To have a non-vanishing contribution, it must anticommute with $S_i$, since Pauli strings either commute or anticommute. In this case, we must have that $A_\beta^{s_i}$ anticommutes with $S_i$, since $A_\beta^{s \setminus s_i}$ trivially commutes with it. However, $A_\beta^{s_i}$ commutes with $S_j$, since they do not overlap. Therefore, since expectation values in product states factorize,
\begin{align}
    & \ev{i[S_i,A_\beta ] }_{\psi_\alpha}\ev{i[S_j,A_\beta ] }_{\psi_\alpha} = \\ 
    &= \ev{i[S_i,A_\beta^{s_i}]} \ev{A_\beta^{s\setminus s_i}} \ev{A_\beta^{s_i}}    \ev{i[S_i,A_\beta^{s \setminus s_i}]}.
\end{align}
This product includes the expectation values of anticommuting Pauli strings:  $A_\beta^{s_i}$ anticommutes with $i\com{S_i,A_\beta^{s_i}}$, since we assumed it anticommutes with $S_i$. As our initial states are stabilizer states, their expectation value cannot be non-zero for a pair of anticommuting Pauli strings, leading to a vanishing contribution to \Eq{average_gram} for \emph{any} Pauli observable $A_\beta$. 

\paragraph*{$2k_h-1$-local configurations suffice.} We have shown that for pair of Hamiltonian terms $S_i, S_j$, it is sufficient to sample local configurations on their shared support; and that it is sufficient to consider Hamiltonian terms that intersect. Therefore, $\frac{1}{\Nspam}K^T K$ becomes diagonal with entries described by \Eq{diagonal_entries} if all the local initial state configurations are uniformly sampled within any $2k_h-1$-ranged patch.

\subsection{Examples for diagonal entry calculations} \label{app:calc-1-loc-obser}

We would like to evaluate the diagonal entries of $\frac{1}{\Nspam} K^T K$ in the limit of all SPAM settings $\Nspam\rightarrow \infty$ [\Eq{all_p_KTK}] for a few simple cases: (i) $1$-local observables and generic Hamiltonian terms, corresponding to the result of \Eq{analytic_spectrum} and (ii) $1$-local Hamiltonian terms and generic observables. %The diagonal entry $\frac{1}{\Nspam} (K^T K)_{ii}$ corresponding to a Hamiltonian term $S_i$ measures the statistical sensitivity of an average experiment to that Hamiltonian term. However, it does not take into account systematic errors as well as correlations between measured observables. 

The diagonal entry for a Hamiltonian term $S_i$ [\Eq{diagonal_entries}] counts the number of measured observables that do not commute with it, normalized by the weight of the observable and of the commutator. For instance, each $1$-local observable $A_\beta$ is compatible with one third of all measurement bases, yielding a normalization factor $\frac{1}{3}$ for the observable locality ($s_\beta=1$). Given a weight $k_i$ Hamiltonian term $S_i$, any non-zero commutator $[S_i, A_\beta]$ is itself of weight $k_i$, vanishing in all but $\frac{1}{3^{k_i}}$ of all initial states ($s_{i,\beta}= k_i$). Finally, for each of the $k_i$ sites in the support of $S_i$ there are two single-site Paulis that differ from it on that site, yielding a combinatorial factor of $2k_i$ non-commuting observables. We conclude that
\begin{equation*} \label{diag KTK with k_b=1}
\frac{1}{\Nspam}(K^T K)_{ii}=\\
\sum_{A_\beta: \com{S_i, A_\beta} \neq 0} \frac{4}{3^{\abs{s_\beta} + \abs{s_{i,\beta}}}}, 
=4\cdot \frac{2 k_i}{3^{k_i+1}}.
\end{equation*}

As another example, we examine the case of a $1$- local Hamiltonian term $k_i=1$ and range $k_\beta$ observables in a one-dimensional chain with periodic boundary conditions. In that case ${\abs{s_\beta} = \abs{s_{i,\beta}}}$. For $k_\beta=2$-ranged observables, 
\begin{equation}
    \frac{1}{\Nspam}(K^T K)_{ii}=\frac{4}{{3^{2+2}}}\cdot {2\cdot2 \cdot3}=4\cdot \frac{4}{3^3}.
\end{equation}
The combinatorial factor consists of three contributions: 2 possible Pauli operators that anticommute with $S_i$ on the site that intersects with it; 2 possible choices of a neighboring site for the observable support; and 3 possible Pauli operators on the neighboring site.
Similarly, for $k_\beta=3$-ranged observables, $\frac{1}{\Nspam}(K^T K)_{ii}=\frac{4}{3^{3+3}}\cdot {2\cdot3 \cdot3^2}+\frac{4}{3^{2+2}}\cdot 2\cdot2\cdot3=4\cdot \frac{18}{3^4}$. 
The left term counts the observables with both range and weight 3, while the right term counts observables with range 3 and weight 2.

The Table below sums results for $\frac{1}{4p}(K^T K)_{ii}$ with up to range 3 observables and weight 3  Hamiltonian terms:
\begin{center}
\begin{tabular}{ |c |c| c| c| }\hline
\backslashbox{$k_i$}{$k_{\beta}$}
& 1 & 2 & 3\\ 
\hline
1 & $\nicefrac{2}{3^2}$ & $\nicefrac{4}{3^3} $& $\nicefrac{18}{3^4}$\\  
\hline
2 & $\nicefrac{4}{3^3}$ & $\nicefrac{16}{3^4}$ & $\nicefrac{52}{3^5}$\\
\hline
3 &$\nicefrac{6}{3^4}$ & $\nicefrac{28}{3^5}$&$\nicefrac{114}{3^6}$\\
\hline
\end{tabular}
\end{center}
The first column and row are the examples calculated above. The columns denote the contribution of all the observables with exact range $k_{\beta}$ to the diagonal entry. In the rows, we consider terms acting nontrivially on each of $k_i$ contiguous sites, such that $k_i$ denotes both the weight and range of the $S_i$ term. This assumption is not necessary for $k_\beta=1$, where the result depends only on the weight of $S_i$.

\subsection{Cyclic measurement scheme in higher-dimensional lattices} \label{app:cyclic}
How many copies of a state $\ket{\psi}$ on a $d$-dimensional lattice do we need to measure each  $k-ranged$ observable at least once?
In $1D$ we defined a unit cell of length $k$. To measure all the $k-ranged$ observables, it sufficed to iterate over the $3^k$ Pauli bases of a single unit cell, repeating measurement bases between different unit cells. Since different unit cell partitions are equivalent up to a permutation, this process yields at least one shot of each obsrvable using exactly $3^k$ copies of the state.

We can generalize this approach to $d$-dimensional lattices. We choose a $d$-dimensional square unit cell of side length $k$, containing $k^d$ sites (see $2D$ case in  \Fig{fig5}). As in the 1D case, we evaluate all the  $k-local$ observables within it using $3^{\left ( k^d\right )}$ different copies, repeating the measurement bases between unit cells. This procedure yields at least one sample of any observable residing within length $k$ cube, even if it does not respect the original lattice partition. Similarly, since there are $6$ Pauli eigenstates per site, $6^{k^d}$ initial states cover all the Pauli eigenstates in each length $k$ cube.

The local measurement process we described improves upon the overlapping tomography protocol of \iRef{Cotler2020QuantumTomography} for the special case of geometrically local lattices. The protocol of \iRef{Cotler2020QuantumTomography} recovers all $k$-local observables (maximal weight $k$) without assuming any spatial structure. The protocol requires $\orderof{\log^2 n}$ copies of a state of $n$ qudits to recover all observables of a fixed locality to a given worst-case error. Our protocol requires an $n$-independent number of copies to recover all short-ranged observables to a constant average accuracy. It effectively eliminates one $\log n$ factor, which is required in order to cover all long-ranged observables. The other $\log n$ factor stems only from the different error metric (worst-case vs. average-case).

\begin{figure}  \includegraphics[width=8.6cm]{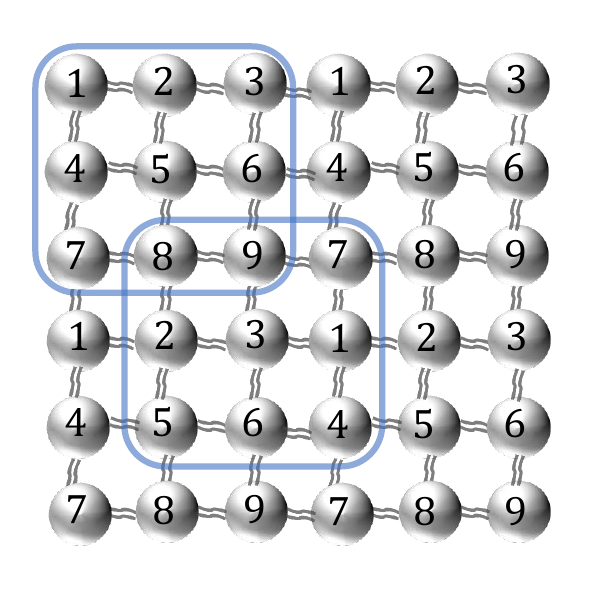} 
  \protect
  \caption{Example of overlapping tomography (cyclic method) in a $d=2$ dimension square lattice with $k_\beta=3$. The boundary of unit cells is colored in blue for two of the unit cells. In this example  we need $3^9$ copies of $\ket{\psi}$ to get one sample of every $3$-ranged observable in the system.}
 \label{fig5}
\end{figure}

\section{XXZ model}\label{app:xxz}
Why does the reconstruction error for the $XXZ$ Hamiltonians of \Eq{XXZ_hamil} decrease more rapidly with the measurement count $\Delta \sim N^{-\nicefrac{1}{3}}$ compared to the $N^{-\nicefrac{1}{4}}$ scaling of generic $2$-local Hamiltonians? 

Since the $XXZ$ Hamiltonians are real in the computational basis, they are invariant under the anti-unitary complex conjugation operator which flips $Y_i \mapsto -Y_i$ \cite{meng-cheng}. As these Hamiltonian contain only even combinations of the other Paulis $X,Z$ as well, they are actually  invariant under a time-reversal symmetry for each axis $\alpha=1,2,3$ of the Bloch sphere, which flips $\sigma^\alpha_i \mapsto -\sigma^\alpha_i$. 

Incidentally, each of the $2$-periodic cyclic initial states we have chosen in \Fig{fig3} is invariant under one of those symmetries, as the spins of each of these states are polarized only within a plane in the Bloch sphere. In other words, each of these states contains only two of the three Paulis, and is invariant under an inversion of the third. We now explain why this implies that for any Pauli observable $A$, either (i) the time derivative vanishes $\partial_t \ev{A} = 0$, or (b) the second time-derivative vanishes $\partial_{tt} \ev{A} = 0$. 

Without loss of generality, we focus on the $Y\mapsto -Y$ time reversal symmetry, and consider a real-valued initial state which is invariant under the standard complex conjugation, such as $\ket{\psi} = \ket{\uparrow \uparrow \uparrow \cdots}$. Each Pauli observable $A$ has a definite parity under this symmetry: it is either real (even) or purely imaginary (odd), depending on the parity of its number of $Y$ operators. Since the Hamiltonian is real, it flips the parity of $A$ under Heisenberg evolution: if $A$ is real, then $\partial_t A = i \com{H,A}$ is imaginary and vice versa. 

Importantly, if an observable $O$ and a state $\rho$ have definite time reversal parities, they must match for the observable to take an expectation value. If $\rho$ is a real density matrix and $O$ an imaginary observable, then
\begin{align*}
    \Tr(\rho O) &= \Tr((\rho O)^T) = \Tr(\rho^T O^T) \\
    &= - \Tr(O \rho) = - \Tr(\rho O),
\end{align*}
since $O^T = (O^\dagger)^* = O^* = -O$, and the trace of any operator is the same as the trace of its transpose. 

Therefore, only time-reversal odd observables $A$ contribute to the reconstruction in a time-reversal even initial state. If $A$ is time-reversal even, then $\partial_t A = i \com{H,A}$ is time-reversal odd; in this case, $\partial_t \ev{A}_{\psi}$ vanishes for a time-reversal even state $\ket{\psi}$ regardless of the particular coefficients of the $XXZ$ Hamiltonian. In contrast, time-reversal odd observables $A$ have a non-trivial time derivative, since $\partial_t A$ is time-reversal even; however, in this case $\partial_{tt} A = -\com{H,\com{H,A}}$ is again time reversal odd, so that $\partial_{tt} \ev{A} = 0$, and the systematic error [\Eq{b_sys}] vanishes to first order in $\delta t$. Therefore, the time-reversal odd observables $A$ that participate in the reconstruction can be measured at later times compared to a generic Hamiltonian, leading to an improved scaling of the reconstruction error $N^{-\nicefrac{1}{3}}$ according to \Eq{optimal-rec-generalized}.

This analysis holds for every Pauli observable in a state whose spins are polarized along the $X-Z$ plane, based on the $Y\mapsto -Y$ symmetry. Since the $XXZ$ Hamiltonian is invariant under inversion of any axis of the Bloch sphere, it generalizes to any initial state whose spins are polarized along a plane. The $2$-periodic initial states satisfy this condition, explaining the findings of \Fig{fig3}. More generally, if a Hamiltonian has time-reversal symmetry along one axis, initial states polarized along the perpendicular plane will yield the improved scaling of the reconstruction error.

%%%%%%%%%%%%%%%%%%%%%%%%%%%%%%%%%%%%%%%%%%%%%%%%%%%%%%%%%%%%%%%%%%%%%

%%%%%%%%%%%%%%%%%%%%%%%%%%%%%%%%%%%%%%%%%%%%%%%%%%%%%%%%%%%%%%%%%%%%%

\end{document}